\documentclass[openacc]{rstransa}

\usepackage{amsfonts}
\usepackage{amsmath}
\usepackage{epigraph}

\newtheorem{definition}{\bf Definition}[section]
\newtheorem{theorem}{\bf Theorem}[section]

\newtheorem{assumption}{\bf Assumption}[section]
\newtheorem{lemma}{\bf Lemma}[section]

\newcommand{\thb}[0]{\operatorname{TH}}

\newcommand{\qstab}[0]{\operatorname{QSTAB}}
\newcommand{\abl}[0]{\operatorname{abl}}
\newcommand{\id}[0]{\mathbb{I}}

\begin{document}


\title{The problem of quantum correlations and the totalitarian principle}


\author{Ad\'an Cabello$^{1,2}$}

\address{$^{1}$Departamento de F\'{\i}sica Aplicada II, Universidad de
Sevilla, E-41012 Sevilla, Spain \\$^{2}$Instituto Carlos~I de F\'{\i}sica Te\'orica y Computacional, Universidad de
Sevilla, E-41012 Sevilla, Spain}


\subject{quantum physics}


\corres{Ad\'an Cabello\\
\email{adan@us.es}}


\begin{abstract}
The totalitarian principle establishes that `anything not forbidden is compulsory'. The problem of quantum correlations is explaining what selects the set of quantum correlations for a Bell and Kochen-Specker (KS) contextuality scenario. Here, we show that two assumptions and a version of the totalitarian principle lead to the quantum correlations. The assumptions are that there is a non-empty set of correlations for any KS contextuality scenario and a statistically independent realisation of any two KS experiments. The version of the totalitarian principle says that any correlation not forbidden by these assumptions can be produced. This paper contains a short version of the proof [presented in {\em Phys. Rev. A} \textbf{100}, 032120 (2019)] and explores some implications of the result.

This article is part of the theme issue `Contextuality and probability in quantum mechanics and beyond'.
\end{abstract}


\maketitle




\section{Introduction}


John Wheeler conjectured that the universe is not `a machine governed by some magic equation', but `a self-synthesizing system', \cite{Wheeler88} in which `[e]verything is built on the unpredictable outcomes of billions upon billions of elementary (\ldots) phenomena' \cite{Wheeler83a}, that are, themselves, `lawless events' \cite{MTZ09}. It is difficult to picture a universe like that. Nevertheless, an interesting intellectual exercise is identifying signatures of such a universe and looking for them in our universe. A possible signature could be that Gell-Mann's totalitarian principle \cite{Gell-Mann56} holds in it. The totalitarian principle states: `anything not forbidden is compulsory'. Of course, we have to clarify what is `anything' and what makes anything `forbidden'.

In this paper we will assume that we are in a Wheelerian lawless universe and we ask ourselves what kind of non-locality and contextuality would be observed in such a universe. Our aim is to address a famous open problem in foundations of physics: identifying the principle that selects the quantum correlations for Bell \cite{Bell64,CHSH69} and Kochen-Specker (KS) contextuality \cite{Specker60,Bell66,KS67} scenarios. 

We assume that the reader has some previous knowledge of Bell non-locality and KS contextuality. However, we will start by reviewing all the concepts needed. In Sec.~\ref{sec:problem}, we define Bell and KS contextuality scenarios, state the problem, and recall the `solution' quantum theory gives. In Sec.~\ref{sec:result}, we present our assumptions and the main result. In Sec.~\ref{sec:proof}, we provide a short proof (an extended version can be found in Ref.~\cite{Cabello18}). In Sec.~\ref{sec:conclusions} we summarize the result. In Sec.~\ref{sec:implications}, we discuss some implications and present some lines for future research.


\section{The problem of quantum correlations}
\label{sec:problem} 


\subsection{Compatibility, non-disturbance and ideal measurements}
\label{sus}


We consider the set of theories that assign probabilities to the outcomes of measurements.
We will denote by $P(x=a | \psi)$ the probability of obtaining outcome $a$ when measuring $x$ on state~$\psi$. By `state' we mean the object that encodes the expectations about the outcomes of future measurements. We do not assume any particular mathematical representation for the states, measurements, and outcomes.

\begin{definition}
	A measurement $z$ with outcomes $c \in C$ is a coarse-graining of a measurement $x$ with outcomes $a \in A$ if, for all $c \in C$, there is $A_c \subseteq A$ such that, for all states $\psi$, 
\begin{equation}
P(z=c | \psi) = \sum_{a \in A_c} P(x=a | \psi)
\end{equation}
and $A_c \cap A_{c'} = \emptyset$ if $c \neq c'$.
\end{definition}

\begin{definition}
Two measurements are compatible if they are coarse-grainings of the same measurement.
\end{definition}

\begin{definition}
Two sets of measurements, $X=\{x_i\}$, with respective outcomes $a_i \in A_i$, and $Y=\{y_j\}$, with respective outcomes $b_j \in B_j$, such that every pair $(x_i,y_j)$ are compatible, are mutually non-disturbing if, for all $x_i \in X$, $a_i \in A_i$, and $y_j,y_k \in Y$,
\begin{equation}
\label{no-disturbance1}
\sum_{b_j \in B_j} P(x_i=a_i,y_j=b_j|\psi)= \sum_{b_k \in B_k} P(x_i=a_i,y_k=b_k|\psi), 
\end{equation}
and, for all $y_i \in Y$, $b_i \in B_i$, and $x_j,x_k \in X$,
\begin{equation}
\label{no-disturbance2}
\sum_{a_j \in A_j} P(x_j=a_j,y_i=b_i|\psi)= \sum_{a_k \in A_k} P(x_k=a_k,y_i=b_i|\psi).
\end{equation} 
Therefore, the marginal probabilities $P(x_i=a_i|\psi)$ are independent of the choice of $y_j \in Y_j$ and the marginal probabilities $P(y_i=b_i|\psi)$ are independent of the choice of $x_j \in X_j$.
\end{definition}

\begin{definition}
A measurement is ideal (or sharp \cite{CY14}) if:
\begin{itemize}
	\item[(i)] it gives the same outcome when performed consecutive times on the same physical system, 
	\item[(ii)] it does not disturb compatible measurements and 
	\item[(iii)] all its coarse-grainings satisfy (i) and (ii).
\end{itemize}
\end{definition}


\subsection{Bell scenarios}


A Bell experiment \cite{Bell64,CHSH69} involves two or more spatially separated agents (typically referred to as parties). Each party performs, on a different subsystem of a composite system, one measurement freely chosen from a fixed set. The choice of measurement of each party is space-like separated from the measurement outcomes observed by the other parties. Therefore, assuming that faster-than-light communication is impossible, measurements performed in the same round of a Bell experiment by different parties are mutually non-disturbing.

A Bell scenario is characterized by the number of parties, the number of measurements each party can perform, the number of outcomes of each measurement, and the relations of compatibility between the measurements. For example, the Clauser-Horne-Shimony-Holt (CHSH) or $(2,2,2)$ Bell scenario \cite{Bell64,CHSH69} has two parties, each of them with two measurements, each of them with two outcomes, and every measurement of one party is compatible with every measurement of the other party.

The simplest quantum systems producing deviations from local realism in Bell scenarios are pairs of qubits. Quantum deviations from local realism require entangled states \cite{Gisin91} and incompatible local measurements. 


\subsection{Kochen-Specker contextuality scenarios}


KS contextuality scenarios (hereafter KS scenarios for brevity) extend Bell scenarios to situations in which compatible measurements are not necessarily spacelike separated as the sense of Bell experiments. A KS contextuality scenario is characterized by a set of ideal measurements, their outcomes, and their relations of compatibility.

The restriction to ideal measurements in KS scenarios (restriction that does not exist in Bell scenarios) makes that, in KS scenarios, compatible measurements become mutually non-disturbing (as occurs in Bell scenarios). Any Bell scenario with ideal measurements is a KS scenario, but Bell scenarios with non-ideal measurements are not KS scenarios.
 
The assumptions of repeatability of the outcomes [i.e., condition~(i) in the definition of ideal measurement] and mutual non-disturbance between measurements that are compatible [i.e., condition~(ii)] justify the assumption of outcome non-contextuality for ideal measurements made for hidden variable theories in Refs.~\cite{Specker60,Bell66,KS67,KCBS08,Cabello08,BBCP09,YO12,KBLGC12}. If conditions~(i) and (ii) do not hold, then the assumption of outcome non-contextuality is not satisfied by classical systems \cite{Spekkens14}.

Every KS set of quantum measurements (as defined in, e.g., Refs.~\cite{KS67,Peres91,CEG96,LBPC14}) defines a KS scenario. There are also KS scenarios that do not require entire KS sets, as quantum deviations from the predictions of KS non-contextual hidden variable theories occur in scenarios with fewer measurements. Among the latter, there are deviations that occur for particular states \cite{KCBS08} and deviations that occur for any state of a given dimension \cite{YO12,KBLGC12,BBC12}. 

Quantum deviations of KS non-contextual hidden variable theories require quantum systems of dimension three (qutrits) or larger \cite{Bell66,KS67}, and do not require entangled states. The simplest KS scenario in which qutrits produce KS contextuality is the Klyachko-Can-Binicio\u{g}lu-Shumovsky (KCBS) KS scenario \cite{KCBS08}, involving five measurements $x_i$, with $i=1,\ldots,5$ (with two possible outcomes), such that $x_i$ and $x_{i + 1}$ (with the sum taken modulo five) are compatible.


\subsection{Contexts and graphs of compatibility}


\begin{definition}
A context in a Bell or KS scenario $S$ is a subset of the measurements in~$S$ which only contains compatible (and mutually non-disturbing) measurements.
\end{definition}

The relations of compatibility between the measurements in $S$ can be pictured by a graph in which vertices represent measurements and edges relations of compatibility. A graph with this interpretation is called a graph of compatibility. For example, the graph of compatibility of the CHSH Bell scenario is a square and the graph of compatibility of the KCBS KS scenario is a pentagon. 

In a graph of compatibility, contexts are represented by cliques. A clique is a set of vertices every pair of which are adjacent.


\subsection{Mutually exclusive events and graphs of exclusivity}


The events of a Bell or KS scenario $S$ and their relations of mutual exclusivity are determined by the measurements, outcomes, and relations of compatibility between the measurements available in~$S$. 

\begin{definition}
Two events of $S$ are mutually exclusive when there is a measurement $M$, defined using the measurements in~$S$, such that each event corresponds to a different outcome of~$M$.
\end{definition}
For example, in the CHSH Bell scenario, the events are of the form $(x=a,y=b|\psi)$ with $x,y,a,b \in \{0,1\}$. Events $(x=a,y=b|\psi)$ and $(x'=a',y'=b'|\psi)$ are mutually exclusive if $x=x'$ and $a \neq a'$ or $y=y'$ and $b \neq b'$. For example, if $x=x'$, $a \neq a'$, and $y \neq y'$, then $M$ is the four-outcome measurement producing events $(x=a,y=0|\psi)$, $(x=a,y=1|\psi)$, $(x=a',y'=0|\psi)$, and $(x=a',y'=1|\psi)$.

The relations of mutual exclusivity between events can be pictured by a graph in which vertices represent events and edges relations of mutual exclusivity. A graph with this interpretation is called a graph of exclusivity \cite{CSW10,CSW14}. 

There are two types of graphs of exclusivity. On the one hand, there are graphs in which the vertices and edges encode (using colours) the measurements and outcomes that define the events and explain why some pairs of events are mutually exclusive. For example, the graph of exclusivity of this type for the 16~events of the CHSH Bell scenario is shown in Fig.~\ref{Fig1}.

On the other hand, there are graphs of exclusivity in which vertices and edges represent abstract events and relations of mutual exclusivity, respectively, without reference to any particular scenario. The vertices of these graphs have no labellings (except, possibly, a numeration). For example, the graph of exclusivity in Fig.~\ref{Fig2} is of this type. 


\begin{figure}[t]
\includegraphics[width=11.4cm]{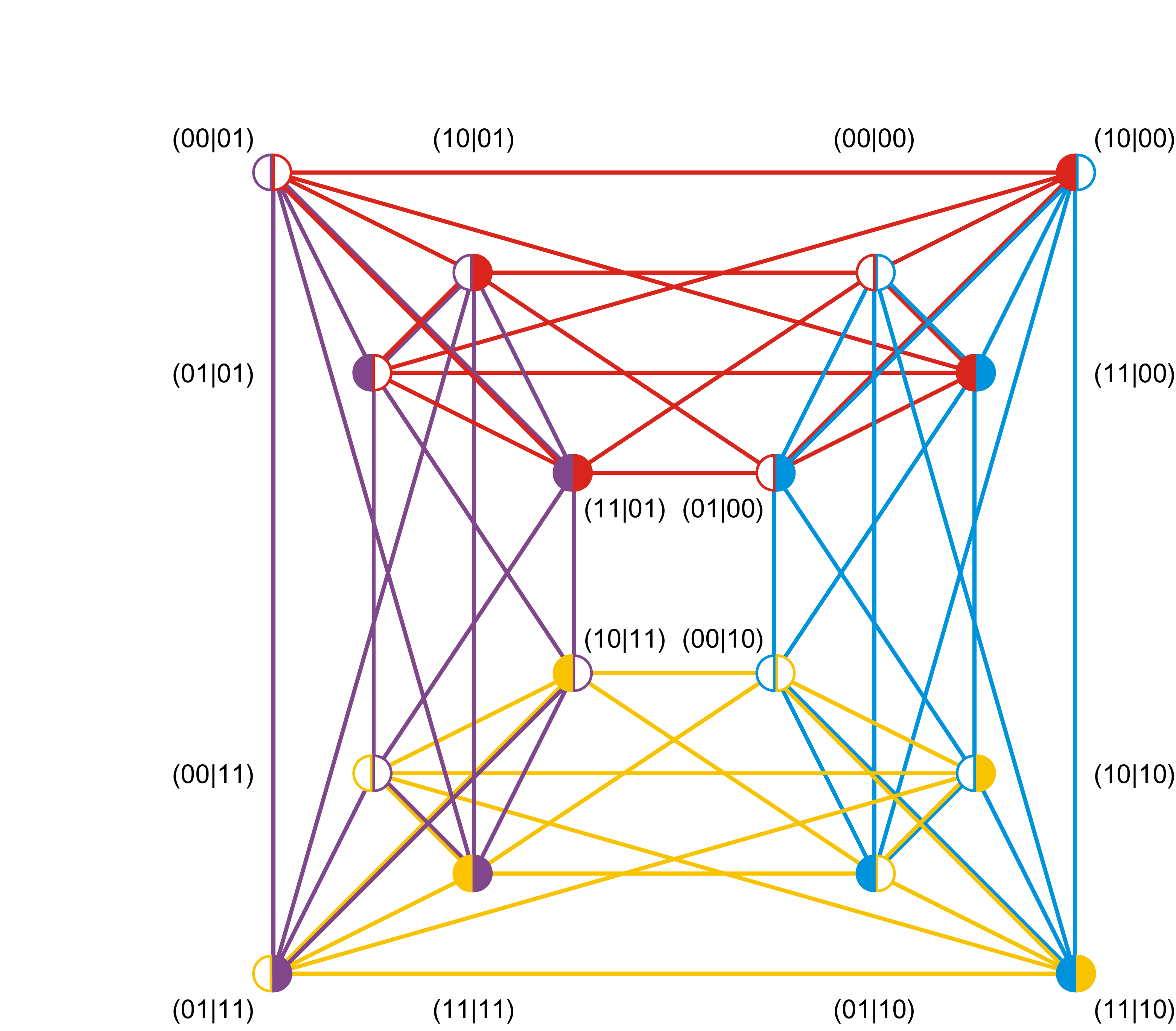}
\caption{Graph of exclusivity $G_{\rm CHSH}$ of the $16$ events of the CHSH Bell scenario. $(ab|xy)$ denotes the event $(x=a,y=b|\psi)$. Each colour corresponds to one of the measurements: red if $x$ is $0$, yellow if $x$ is $1$, cyan if $y$ is $0$ and purple if $y$ is $1$. Each event is characterized by the outcomes of two measurements. This is the reason why each vertex has two half circumferences. An empty (full) half circumference of a given colour denotes that the outcome of the measurement represented by that colour is $0$ (respectively, $1$).}
\label{Fig1}
\end{figure}


\subsection{Correlations}


What we informally refer to as `correlations' for a particular Bell or KS scenario~$S$ is a set $\mathbf{p}(S) \in \cal {P}(S)$ of probability distributions, one for each context. Following the terminology introduced in \cite{Tsirelson93} (and used in, e.g., \cite{BCPSW14}), we will refer to every $\mathbf{p}(S)$ as a `behaviour' for~$S$. In the literature $\mathbf{p}(S)$ is also called an `empirical model' \cite{AB11} or `probability model' \cite{AH12,AFLS15} for~$S$.

For a given Bell or KS scenario $S$, every initial state and set of measurements produce a behaviour. For example, for the CHSH Bell scenario $S_{\rm CHSH}$, if, for simplicity, we denote probabilities $P(x=a,y=b|\psi)$ as $P(ab|xy)$, a behaviour can be represented by the matrix
\begin{equation}
\label{beha}
\mathbf{p}(S_{\rm CHSH}) =
\begin{bmatrix}
P(00|00) & P(01|00) & P(10|00) & P(11|00) \\
P(00|01) & P(01|01) & P(10|01) & P(11|01) \\
P(00|10) & P(01|10) & P(10|10) & P(11|10) \\
P(00|11) & P(01|11) & P(10|11) & P(11|11)
\end{bmatrix}.
\end{equation}
Note that each row in the right-hand side of Eq.~(\ref{beha}) contains the probabilities of the events of one context. Alternatively, $\mathbf{p}(S_{\rm CHSH})$ can be seen as an assignment of probabilities to the (suitably ordered) vertices of the graph of exclusivity $G_{\rm CHSH}$ in Fig.~\ref{Fig1}. In this case, we will represent $\mathbf{p}(S_{\rm CHSH})$ by a vector $\mathbf{p}(G_{\rm CHSH}) \in [0,1]^{|V(G_{\rm CHSH})|}$, where $V(G_{\rm CHSH})$ is the set of vertices of $G_{\rm CHSH}$.


For a given Bell (KS scenario) $S$, the set of local realistic (KS non-contextual) behaviours is a closed convex polytope, i.e., a closed convex set whose boundaries are flat, called the local polytope (KS non-contextual polytope). The corresponding set in quantum theory is a convex sets that, in general, is not a polytope. The local polytope for the CHSH Bell scenario is identical to the non-contextual polytope for the KS scenario involving four two-outcome ideal measurements and having a square as graph of compatibility. Similarly, the local polytope for any Bell scenario is identical to the non-contextual polytope for the KS scenario that has the same number of measurements, outcomes, and graph of compatibility.


\subsection{Constraints for behaviours for Bell and Kochen-Specker scenarios}


For a fixed Bell or KS scenario~$S$, every behaviour $\mathbf{p}(S)$ must satisfy three constraints:
\begin{itemize}
\item[(A)] Normalization: For every context $\{x,\ldots,z\}$ (with respective outcomes $a \in A,\ldots,c \in C$) in $S$, 
\begin{equation}
\sum_{a \in A,\ldots,c \in C} P(x=a,\ldots,z=c|\psi)=1.
\end{equation}
Any subset of a context is also a context.
\item[(B)] Non-disturbance: Every pair $(X,Y)$ of mutually non-disturbing sets of measurements in~$S$ must satisfy Conditions~(\ref{no-disturbance1}) and (\ref{no-disturbance2}). 
\item[(C)] The probability of each event of~$S$ must only be a function of the state and measurement outcomes that define the event. For example, $P(x=a,y=b|\psi)$ must only be a function of $\psi$, $x=a$, and $y=b$. This constraints the behaviours since, for a fixed~$S$, any behaviour must correspond to a fixed $\psi$ and a fixed set $\{x=a,y=b\}$.
\end{itemize}


\subsection{Quantum behaviours for Bell and Kochen-Specker scenarios}
\label{Nai}


Any behaviour $\mathbf{p}(S)$ for a Bell or KS scenario~$S$ allowed by quantum theory satisfies the following conditions:
\begin{itemize}
\item[(I)] The initial state $\psi$ of the system can be associated with a vector with unit norm $| \psi \rangle$ in a Hilbert space ${\mathcal H}$. 
\item[(II)] The state after performing any set $\{x^{(i)}\}$ of compatible measurements in $S$ and obtaining, respectively, outcomes $\{a_i\}$ can be associated with a vector with unit norm 
\begin{equation}
\label{2.6}
|\psi' \rangle = N_i \prod_i E^{(i)}_{a_i} | \psi \rangle,
\end{equation}
where $N_i$ is a normalisation constant and $\{E^{(i)}_{a_i}\}$ are projection operators acting on~${\mathcal H}$. The projection operator $E^{(i)}_{a_j}$ corresponds to measurement $x^{(i)}$ with outcome~$a_j$. The projection operators corresponding to different outcomes of the same measurement satisfy
\begin{equation}
E^{(i)}_{a_j} E^{(i)}_{a_k} = \delta_{j,k} E^{(i)}_{a_k},
\end{equation} 
and 
\begin{equation}
\sum_k E^{(i)}_{a_k}=\id,
\end{equation}
where $\id$ is the identity operator.
If $x^{(i)}$ and $x^{(k)}$ are compatible measurements, then 
\begin{equation}
[E^{(i)}_{a_j},E^{(k)}_{a_m}]=0\; \forall j,m,
\end{equation}
where $[\ldots]$ denotes the commutator.
\item[(III)] The probability of obtaining $\{a_i\}$ when measuring $\{x^{(i)}\}$ on state $\psi$ satisfies $| \langle \psi' | \psi \rangle |^2$, where $|\psi' \rangle$ is given by Eq.~(\ref{2.6}).
\end{itemize}

Remarkably, the characterization of the quantum behaviours is similar for Bell scenarios (in which we do not assume that measurements are ideal) and for KS scenarios (in which all measurements are ideal by definition). This reflects the fact that any quantum behaviour for a Bell scenario can be attained with ideal measurements. This follows from Neumark's (also named Naimark's) dilation theorem \cite{Neumark40a,Neumark40b,Neumark43,Holevo80,Peres95} that states that every generalized measurement in quantum theory [represented by a positive-operator valued measure (POVM)] can be implemented as an ideal quantum measurement [represented by a projection-valued measure (PVM)] on a larger Hilbert space. In a Bell scenario, any local POVM $x$ admits a local dilation to a PVM that is common to every context in which $x$ appears. Due to this, the set of quantum behaviours for the Bell CHSH scenario is identical to set of quantum behaviours for the KS~scenario involving four two-outcome ideal measurements and having a square as graph of compatibility. And similarly for every Bell scenario.


\subsection{The problem of quantum correlations}


The problem we address here is identifying the physical reason or principle that explains why the behaviours that are realisable for Bell and KS scenarios are those that satisfy conditions~(I)--(III). The question we want to answer is where does the `irrational effectiveness' \cite{Wigner60} of the Hilbert space formalism supplemented with Born's rule for singling out the physically realisable behaviours come from, why the quantum formalism is empirically successful.

None of the proposed principles (non-signalling \cite{PR94}, non-triviality of communication complexity \cite{vanDam99}, information causality \cite{PPKSWZ09}, macroscopic locality \cite{NW09}, exclusivity \cite{Cabello13}, and local orthogonality \cite{FSABCLA13}) has succeeded in selecting the quantum behaviours even in the simplest Bell scenario. In fact, for every non-trivial Bell scenario, there are non-quantum behaviours which seem to satisfy all these principles \cite{NGHA15}.


\section{Result}
\label{sec:result}


\subsection{Assumptions}


We make the following assumptions:
\begin{assumption}
\label{Ass:BKS}
There is a non-empty set of behaviours for any KS scenario.
\end{assumption}

\begin{assumption}
\label{Ass:Independence}
There is a statistically independent joint realization of any two KS~experiments.
\end{assumption}

\begin{definition}
Two experiments ${\cal A}$ and ${\cal B}$ are statistically independent if
the occurrence of any of the events of ${\cal A}$ (${\cal B}$) does not affect the probability of occurrence of any of the events of ${\cal B}$ (respectively, ${\cal A}$).
\end{definition}

Consequently, if ${\cal A}$ and ${\cal B}$ are statistically independent experiments and matrix $\mathbf{p}(S_{\cal A})$ is the behaviour for ${\cal A}$ [e.g. $\mathbf{p}(S_{\cal A})$ could be the one in Eq.~(\ref{beha})] and $\mathbf{q}(S_{\cal B})$ is the behaviour for ${\cal B}$, then an observer can define an experiment $({\cal A},{\cal B})$ with a behaviour given by $\mathbf{p}(S_{\cal A}) \otimes \mathbf{q} (S_{\cal B})$, where $\otimes$ is the tensor product. 


\subsection{Result}


The aim of this paper is to prove the following 
\begin{theorem}
\label{Th:scenario}
The set of behaviours ${\cal P}(S)$ allowed by quantum theory for any Bell or KS scenario~$S$ is equal to the largest set allowed by assumptions~\ref{Ass:BKS} and~\ref{Ass:Independence}.
\end{theorem}


\subsection{The totalitarian principle}


The fact that the quantum set equals the largest set allowed by assumptions~\ref{Ass:BKS} and~\ref{Ass:Independence} implies that any behaviour allowed by these assumptions is compulsory in the sense that there exist a preparation and a set of measurements that produce it. This resembles Gell-Mann's totalitarian principle \cite{Gell-Mann56}: `anything not forbidden is compulsory'. 

The story of the totalitarian principle can be summarized as follows. In the context of the interactions between baryons, antibaryons and mesons, Gell-Mann made the observation that any process which is not forbidden by a conservation law is not only allowed but {\em must} be included in the sum over all paths which contribute to the outcome of the interaction. Gell-Mann name it the `principle of compulsory strong interactions' \cite{Gell-Mann56} and commented that `is related to the state of affairs that is said to prevail in a perfect totalitarian state. Anything that is not compulsory is forbidden' \cite{Gell-Mann56}. After that, the principle started to be called `Gell-Mann's totalitarian principle' \cite{BS69,Zweig10} and reformulated as `anything not forbidden is compulsory'. 
However, Trigg \cite{Trigg70} and Weinberg \cite{Weinberg05} have pointed out that the author who deserves the credit for the totalitarian principle is T.~H.~White because, as Trigg remarks, `[i]n {\em The Sword in the Stone}, part~1 of {\em The Once and Future King}, Wart, the character who will later be King Arthur, is being educated by Merlin by being transformed into various animals. One of his experiences is as an ant, and he finds that the ant hill is run on the totalitarian principle' \cite{Trigg70}. Specifically, in Chapter~XIII of {\em The Sword in the Stone} \cite{White38}, 
`EVERYTHING NOT FORBIDDEN IS COMPULSORY' is the slogan carved over the entrance to each tunnel in the ant fortress.

In our version of the totalitarian principle, `anything' refers to `any behaviour for a Bell or KS scenario' and `forbidden' means `forbidden by assumptions~\ref{Ass:BKS} and~\ref{Ass:Independence}'.


\section{Proof}
\label{sec:proof}


A behaviour for scenario $S$ satisfies the exclusivity principle (EP) \cite{CSW10,Cabello13,Yan13,CSW14,ATC14,CY14,Cabello15,Henson15,AFLS15} if the sum of the probabilities of the events of any set in which every two events are mutually exclusive in~$S$ is bounded by~$1$.

The proof of Theorem~\ref{Th:scenario} begins with the following observations.

\begin{lemma}
	\label{Lem:1}
	The behaviours for any KS scenario must satisfy the EP. 
\end{lemma}

\begin{proof}
In KS scenarios all measurements are ideal. It can be proven \cite{CY14,Cabello18} that the~EP holds for ideal measurements. Therefore, every behaviour for any KS scenario must satisfy the~EP.
\end{proof}

\begin{lemma}
	\label{Th:3}
	The behaviours for any bipartite Bell scenario must satisfy the EP.
\end{lemma}

\begin{proof}
	By definition of Bell scenario, behaviours must satisfy normalization and non-signalling (i.e., non-disturbance). It can be proven \cite{CSW10,Cabello18} that, for bipartite Bell scenarios, the set of behaviours satisfying the EP (applied to a single copy) is equal to the set of behaviours that satisfy normalization and the non-signalling principle.
\end{proof}

Assumption~\ref{Ass:Independence} assures that there is a statistically independent joint realization of any two KS experiments ${\cal A}$~and~${\cal B}$ (including Bell experiments). This allows us to define experiments of the type $({\cal A},{\cal B})$ described before.

If ${\cal A}$ and ${\cal B}$ are KS experiments, then $({\cal A},{\cal B})$ can be seen as a single KS experiment. Therefore, lemma~\ref{Lem:1} assures that the behaviours for $({\cal A},{\cal B})$ must satisfy the~EP. 

If ${\cal A}$ and ${\cal B}$ are bipartite Bell experiments, then $({\cal A},{\cal B})$ can be seen as a bipartite Bell experiment. Therefore, lemma~\ref{Th:3} assures that the behaviours for $({\cal A},{\cal B})$ must satisfy the~EP.

If ${\cal A}$ is a KS experiment and ${\cal B}$ a bipartite Bell experiment, then $({\cal A},{\cal B})$ can be seen as a bipartite Bell experiment. Therefore, lemma~\ref{Th:3} assures that the behaviours for $({\cal A},{\cal B})$ must satisfy the~EP.

Similar arguments apply to $n$~statistically independent experiments. In particular, the same behaviour $\mathbf{p}(S_A)$ for a Bell or KS~experiment can be composed with itself as many times as we want. If $\mathbf{p}(S_A)$ occurs in $n$~statistically independent experiments ${\cal A}_i$, then, by lemmas~\ref{Lem:1} and~\ref{Th:3}, the corresponding behaviour for the corresponding experiment $({\cal A}_1, \ldots, {\cal A}_n)$, that is, $\mathbf{p}(S_{(A_1,\ldots,A_n)})=\mathbf{p}(S_{\cal A})^{\otimes n}$, where $\mathbf{p}(S_{\cal A})^{\otimes n}$ denotes the tensor product of $n$ copies of $\mathbf{p}(S_{\cal A})$, must satisfy the EP for any $n$.

Different sets of events may share the same graph of exclusivity~$G$. This leads to the following definition:

\begin{definition}
The set ${\cal P}(G)$ of assignments of probabilities to the vertices of graph~$G$ is the set of vectors $\mathbf{p}(G) \in [0,1]^{|V(G)|}$ such that the components of $\mathbf{p}(G)$ are the probabilities of $|V(G)|$ events with graph of exclusivity~$G$ in a behaviour for {\em some} Bell or KS scenario.
\end{definition} 
That is, ${\cal P}(G)$ contains the vectors of probabilities with $|V(G)|$ components corresponding to events that have~$G$ as graph of exclusivity produced in {\em all} Bell or KS scenarios.

\begin{lemma}
	\label{Lem:2}
	For any self-complementary graph of exclusivity $G$, the theta body of $G$, $\thb(G)$, is the largest set of assignments of probabilities ${\cal P}(G)$ such that every $\mathbf{p}(G) \in {\cal P}(G)$ satisfies the EP applied to any number of independent copies of $\mathbf{p}(G)$ and such that $\mathbf{p}(G) \otimes \mathbf{q}(G)$ satisfies the EP for every $\mathbf{p}(G), \mathbf{q}(G) \in {\cal P}(G)$.
\end{lemma}

Given a graph $G$, $\overline{G}$ denotes the complement of $G$, i.e. the graph with the same vertices as $G$ and such that two distinct vertices of $\overline{G}$ are adjacent if and only if they are not adjacent in $G$. A graph $G$ is self-complementary if $G$ and $\overline{G}$ are isomorphic. The theta body of $G$, denoted $\thb(G)$, is a well-studied convex set in graph theory (e.g. \cite{GLS88,Knuth94}). It was introduced in \cite{GLS86} and, among the many ways to express it, the following one, using Dirac's notation, is particularly useful for our purposes:
\begin{equation}
\begin{split}
\thb(G) = & \{ \mathbf{p}(G) \in [0,1]^{|V(G)|} : p_i = |\langle x^{(i)}_{a_i} \psi | \psi \rangle|^2,\;\\
& |\langle \psi| \psi \rangle| = 1,\; |\langle x^{(i)}_{a_i}\psi | x^{(i)}_{a_i}\psi \rangle| = 1,\;\\
&\langle x^{(i)}_{a_i}\psi | x^{(j)}_{a_j}\psi \rangle = 0,\;\forall (i,j) \in E(G) \},
\end{split}
\end{equation}
where $E(G)$ is the set of edges of $G$.

\begin{proof}(of Lemma~\ref{Lem:2})
There are two necessary conditions that any candidate for ${\cal P}(G)$ satisfying assumptions~\ref{Ass:BKS} and~\ref{Ass:Independence} must satisfy:

Condition 1: Any $\mathbf{p}(G) \in {\cal P}(G)$ must satisfy the EP applied to any number $n$ of independent copies of $\mathbf{p}(G)$. Therefore, for any $n$, 
\begin{equation}
\label{Cond1}
{\cal P}(G) \subseteq {\cal E}^n(G),
\end{equation}
where
\begin{equation}
{\cal E}^n(G) = \{\mathbf{p}(G) \in [0,1]^{|V(G)|} : \mathbf{p}(G)^{\,\otimes n} \in \qstab(G^{\ast n})\}, 
\end{equation}
where $\mathbf{p}(G)^{\,\otimes n}=\mathbf{p}(G) \otimes \cdots \otimes \mathbf{p}(G)$ ($n$ times), $G^{\ast n}$ is the OR product of $n$ copies of $G$, and 
\begin{equation}
\qstab(G) = \{ \mathbf{p}(G) \in [0,1]^{|V(G)|} : \sum_{i\in c} p_i \le 1\,\;\forall c \in C(G) \},
\end{equation}
where $C(G)$ is the set of cliques of $G$.
The OR product (also called disjunctive or co-normal product) of $G$ and $G'$, denoted $G \ast G'$, is the graph with $V(G \ast G') = V(G) \times V(G')$ and $((i, i'), (j, j')) \in E(G \ast G')$ if and only if $(i,j) \in E(G)$ or $(i',j') \in E(G')$.
$\qstab(G)$ is a famous convex set in graph theory called the clique-constrained stable set polytope \cite{GLS88} (or fractional stable set polytope \cite{GLS88} or fractional vertex packing polytope \cite{GLS86}) of $G$. 

Condition~2: Any $\mathbf{p}(G), \mathbf{q}(G) \in {\cal P}(G)$ must satisfy the EP applied to one copy of $\mathbf{p}(G)$ and one independent copy of $\mathbf{q}(G)$. 
This implies that
\begin{equation}
\label{Cond2}
{\cal P}(G) \subseteq \abl[{\cal P}(G)],
\end{equation} 
where $\abl[{\cal P}(G)]$ is the antiblocker of ${\cal P}(G)$, defined as 
\begin{equation}
\abl[{\cal P}(G)]= \{\mathbf{q}(G) \geq \mathbf{0} : \mathbf{p}(G) \cdot \mathbf{q}(G) \leq 1\,\forall \mathbf{p}(G) \in {\cal P}(G) \},
\end{equation}
where $\cdot$ is the dot product \cite{GLS88,Knuth94,GLS86}.

If $G$ is self-complementary, then, for any $n$,
\begin{equation}
\label{Cond22}
\abl[{\cal E}^{n-1}(G)] \subseteq \abl[{\cal E}^{n}(G)] \subseteq {\cal E}^{n}(G).
\end{equation} 
Therefore, (\ref{Cond1}) implies that
\begin{equation}
\label{Cond11}
{\cal P}(G) \subseteq \lim_{n \to \infty } {\cal E}^n(G).
\end{equation}
and (\ref{Cond2}) implies that
\begin{equation}
\label{Cond3}
{\cal P}(G) \subseteq \lim_{n \to \infty } \abl[{\cal E}^n(G)].
\end{equation}
If $G$ is self-complementary, as $n$~tends to infinity, ${\cal E}^{n}(G)$ and $\abl[{\cal E}^{n}(G)]$ tend to $\thb(G)$ from above and below, respectively, and $\thb(G)$ is the largest set satisfying (\ref{Cond11}) and (\ref{Cond3}) [the other sets satisfying these conditions are subsets of $\thb(G)$]. Therefore, if $G$ is self-complementary, the largest ${\cal P}(G)$ is $\thb(G)$. In this case, ${\cal P}(G) = \abl[{\cal P}(G)]$ \cite{GLS86,GLS88,Knuth94}.
\end{proof}

Lemma~\ref{Lem:2} states that for any self-complementary graph of exclusivity $G$, the largest set of assignments of probabilities satisfying Conditions~1 and~2 is the one that contains all those assignments of probabilities that satisfy that:
\begin{itemize}
\item[(I')] The state of the system can be associated with a vector with unit norm $| \psi \rangle$ in a vector space ${\mathcal V}$ with an inner product.
\item[(II')] The state after performing, on state $\psi$, a measurement $x^{(i)}$ and obtaining outcome $a_i$ can be associated with a vector with unit norm $ | x^{(i)}_{a_i}\psi \rangle$ in ${\mathcal V}$. Post-measurement states corresponding to mutually exclusive events are associated to mutually orthogonal vectors.
\item[(III')] The probability of the event $(x^{(i)}_{a_i} | \psi)$ can be obtained as $|\langle x^{(i)}_{a_i} \psi | \psi \rangle|^2$.
\end{itemize}

Interestingly, in quantum theory, for any $G$, ${\cal P}(G)=\thb(G)$ \cite{CSW14}. Therefore, at least for self-complementary graphs of exclusivity, assumptions~\ref{Ass:BKS} and~\ref{Ass:Independence} select the quantum set.

\begin{lemma}
\label{Lem:3}
For any graph of exclusivity $G$, $\thb(G)$ is the largest set ${\cal P}(G)$ of assignments of probabilities such that every $\mathbf{p}(G) \in {\cal P}(G)$ satisfies the EP applied to any number of independent copies of $\mathbf{p}(G)$ and such that $\mathbf{p}(G) \otimes \mathbf{q}(G)$ satisfies the EP for every $\mathbf{p}(G), \mathbf{p}(G) \in {\cal P}(G)$.
\end{lemma}

\begin{proof}
For any $G$, there is an operation that maps $G$ into a a self-complementary graph $H(G)$ such that the largest ${\cal P}[H(G)]$ allowed by assumptions~\ref{Ass:BKS} and \ref{Ass:Independence} determines the largest ${\cal P}(G)$ allowed by these assumptions. This map can be visualized as follows.
Consider an experiment~${\cal E}$ producing $n$ events $\{e_k\}_{k=1}^n$ whose graph of exclusivity is $G$. Then, consider three additional mutually independent experiments: experiment~${\cal X}$, producing events $\{x_k\}_{k=1}^n$ whose graph of exclusivity is $\overline{G}$, experiment~${\cal Y}$, producing events $\{y_k\}_{k=1}^n$ whose graph of exclusivity is $\overline{G}$, and experiment~${\cal Z}$ producing events $\{z_k\}_{k=1}^n$ whose graph of exclusivity is~$G$. 
Now suppose an observer witnessing ${\cal E}$, ${\cal X}$, ${\cal Y}$ and ${\cal Z}$. Suppose that this observer has 
three independent coins ${\cal A}$, ${\cal B}$ and ${\cal C}$, each of them producing two mutually exclusive events: ${\cal A}$ producing events $a_0$ or $a_1$, ${\cal B}$ producing $b_0$ or $b_1$, and ${\cal C}$ producing $c_0$ or $c_1$. Suppose that, using the four independent experiments and the three coins, the observer defines the following $4 n$ events: $\{(a_0,e_k),(a_1,b_0,x_k),(b_1,c_0,y_k),(c_1,z_k)\}_{k=1}^n$, where, e.g. $(a_0,e_1)$ is the event in which coin ${\cal A}$ gives $a_0$ and experiment ${\cal E}$ gives $e_1$. $H(G)$ is the graph of exclusivity of these $4 n$ events. Fig.~\ref{Fig2} shows $H(G)$ when $G$ is an heptagon.


\begin{figure}[t]
\includegraphics[width=13.6cm]{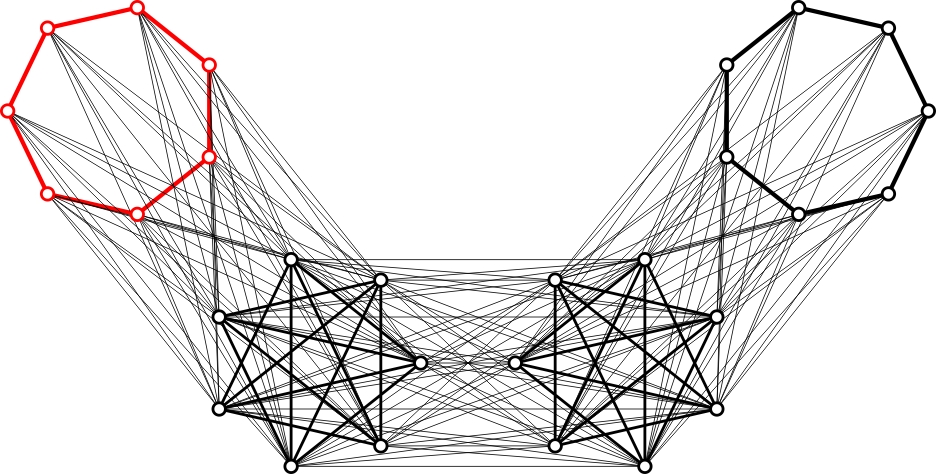}
\caption{Graph $H(G)$ used in the proof of lemma~\ref{Lem:3}, when $G=C_7$ (the heptagon in red). All the vertices of the heptagon are connected to all the vertices of its complement, which are also connected to all the vertices of a second copy of the complement, which are connected to a second heptagon. The resulting graph $H(C_7)$ is self-complementary. In addition, the largest set ${\cal P}[H(C_7)]$ allowed by assumptions~\ref{Ass:BKS} and \ref{Ass:Independence} determines the largest set ${\cal P}(C_7)$ allowed by these assumptions.}
\label{Fig2}
\end{figure}


There are two properties of $H(G)$ that are important for us. The first one is that any assignment of probabilities $\mathbf{h}[H(G)] \in {\cal P}[H(G)]$ can be implemented by suitably choosing assignments $\mathbf{p}(G) \in {\cal P}(G)$, $\mathbf{x}(\overline{G}) \in {\cal P}(\overline{G})$, $\mathbf{y}(\overline{G}) \in {\cal P}(\overline{G})$ and $\mathbf{z}(G) \in {\cal P}(G)$ for, respectively, $\{e_k\}_{k=1}^n$, $\{x_k\}_{k=1}^n$, $\{y_k\}_{k=1}^n$ and $\{z_k\}_{k=1}^n$, and assignments of probabilities $\mathbf{a}(K_2) \in {\cal P}(K_2)$, $\mathbf{b}(K_2)\in {\cal P}(K_2)$ and $\mathbf{c}(K_2)\in {\cal P}(K_2)$ for, respectively, $\{a_0,a_1\}$, $\{b_0,b_1\}$ and $\{c_0,c_1\}$ ($K_2$ is the complete graph on two vertices; the graph of exclusivity of the events of tossing a coin) \cite{Cabello18}. Therefore, the largest ${\cal P}(G)$ allowed by assumptions~\ref{Ass:BKS} and \ref{Ass:Independence} can always be obtained from the largest ${\cal P}[H(G)]$ allowed by these assumptions by suitably tracing out its elements. 

The second property is that $H(G)$ is self-complementary. Therefore, by lemma~\ref{Lem:2}, the largest ${\cal P}[H(G)]$ allowed by assumptions~\ref{Ass:BKS} and \ref{Ass:Independence} is $\thb[H(G)]$. 
Using these two properties, we conclude that the largest ${\cal P}(G)$ allowed by assumptions~\ref{Ass:BKS} and \ref{Ass:Independence} is $\thb(G)$. 
\end{proof}


A consequence of lemma~\ref{Lem:3} is that the set of behaviours for any given Bell or KS scenario $S$ is a subset of $\thb(G_S)$, where $G_S$ is the graph of exclusivity of all possible events of~$S$. However, $\thb(G_S)$ may contain behaviours that are forbidden in $S$. To remove them, we need to take into account constraints~(A), (B) and~(C) for~$S$ (see Sec.~\ref{sec:problem}). 

The way constraints~(A) and~(B) for $S$ remove elements of $\thb(G_S)$ is easy to understand. The way constraint~(C) for $S$ excludes behaviours of $\thb(G_S)$ is more subtle. To illustrate it, we will focus on the case that $S$ is the CHSH Bell scenario. It is possible to produce any assignment of probabilities in $\thb(G_S)$ if we have a suitable vector with unit norm $|\psi\rangle$ in a vector space~${\cal V}$ [by Condition~(I')] and $16$~suitable vectors of unit norm (or rank-one projectors) and their orthogonal complements in~${\cal V}$ (since measurements have only two outcomes) [by Condition~(II')], so behaviours satisfy Condition~(III'). However, in the CHSH Bell scenario we have only four two-outcome measurements. Therefore, we have only four pairs projector-orthogonal complement, each of them represented in Fig.~\ref{Fig1} by one edge with a full half circumference in one extreme and an empty half circumference in the other extreme, all of the same colour. According to constraint~(C), these four pieces must be combined exactly as shown in Fig.~\ref{Fig1} to produce the events and explain their relations of mutual exclusivity.

It is not difficult to see that, for any given $S$, constraints~(A), (B) and~(C) applied after conditions~(I'), (II') and~(III') are equivalent to conditions~(I), (II) and~(III) (see Sec.~\ref{Nai}). See Ref.~\cite{Cabello18} for details. This finishes the proof of Theorem~\ref{Th:scenario}.


\section{Conclusion}
\label{sec:conclusion}


We started by assuming that there is a non-empty set of behaviours for any KS scenario and a statistically joint realisation of any two KS experiments. Then we showed that, for any specific Bell or KS scenario, the largest set of behaviours allowed by these assumptions is the set allowed by quantum theory. 

This suggests a simple principle for quantum correlations: every behaviour that is not forbidden by these assumptions is feasible. That is, there is a preparation and a set of measurements that produces it. This result fits with Wheeler's thesis that the universe lacks of laws for the outcomes of some experiments and with Born's intuition that quantum theory is a consequence of the nonexistence of `conditions for a causal evolution' \cite{Born26}. One can still assume that some laws and conditions for a causal evolution exist. However, then our result shows that they have no effect on the correlations allowed for Bell and KS scenarios.


\section{Implications}
\label{sec:implications}


In this final part, we explore some of the implications of Theorem~\ref{Th:scenario} (and of the lemmas used for its proof) and address some frequently asked questions.


\subsection{Specker's triangle parable and Wright's pentagon}


In Ref.~\cite{Specker60}, within the framework of a parable involving an over-protective Assyrian prophet, Specker introduced the following behaviour for a scenario in which there are three measurements $\{1,2,3\}$ whose graph of compatibility is a triangle (denoted as $C_3$) and whose possible outcomes are $0$ and $1$. If we denote as $P(a_i a_j|x_i x_j)$ the probability of obtaining outcomes $a_i$ and $a_j$ when measuring $x_i$ and $x_j$, respectively, a behaviour in this scenario can be represented by a matrix
\begin{equation}
\label{behaS}
\mathbf{p}(S_{C_3}) =
\begin{bmatrix}
P(00|12) & P(01|12) & P(10|12) & P(11|12) \\
P(00|23) & P(01|23) & P(10|23) & P(11|23) \\
P(00|31) & P(01|31) & P(10|31) & P(11|31)
\end{bmatrix}.
\end{equation}
Specifically, Specker's behaviour is 
\begin{equation}
\label{behaS2}
\mathbf{p}_{\rm S}(S_{C_3}) =
\begin{bmatrix}
0 & \frac{1}{2} & \frac{1}{2} & 0 \\[5pt]
0 & \frac{1}{2} & \frac{1}{2} & 0 \\[5pt]
0 & \frac{1}{2} & \frac{1}{2} & 0
\end{bmatrix}.
\end{equation}
It is easy to see that $\mathbf{p}_{\rm S}(S_{C_3})$ violates the EP when we consider two independent copies of $\mathbf{p}_{\rm S}(S_{C_3})$. Therefore, under assumptions~\ref{Ass:BKS} and~\ref{Ass:Independence}, the proof of lemma~\ref{Lem:1} shows that $\mathbf{p}_{\rm S}(S_{C_3})$ cannot be produced with ideal measurements. However, $\mathbf{p}_{\rm S}(S_{C_3})$ can be trivially produced if we remove the assumption that the measurements are ideal \cite{LSW11,LSW17}.

The same argument and conclusion apply to a behaviour proposed by Wright in Ref.~\cite{Wright78} for a scenario in which there are five measurements whose graph of compatibility is a pentagon (denoted as $C_5$). Using a similar notation to the one adopted in the previous example, a behaviour for this scenario can be represented by a matrix
\begin{equation}
\label{behaW}
\mathbf{p}(S_{C_5}) =
\begin{bmatrix}
P(00|12) & P(01|12) & P(10|12) & P(11|12) \\
P(00|23) & P(01|23) & P(10|23) & P(11|23) \\
P(00|34) & P(01|34) & P(10|34) & P(11|34) \\
P(00|45) & P(01|45) & P(10|45) & P(11|45) \\
P(00|51) & P(01|51) & P(10|51) & P(11|51) \\
\end{bmatrix}.
\end{equation}
Specifically, Wright's behaviour is
\begin{equation}
\label{behaS3}
\mathbf{p}_{\rm W}(S_{C_5}) =
\begin{bmatrix}
0 & \frac{1}{2} & \frac{1}{2} & 0 \\[5pt]
0 & \frac{1}{2} & \frac{1}{2} & 0 \\[5pt]
0 & \frac{1}{2} & \frac{1}{2} & 0 \\[5pt]
0 & \frac{1}{2} & \frac{1}{2} & 0 \\[5pt]
0 & \frac{1}{2} & \frac{1}{2} & 0
\end{bmatrix}.
\end{equation}


\subsection{Popescu-Rohrlich boxes}


In Ref.~\cite{PR94}, Popescu and Rohrlich (and previously other authors \cite{Rastall85,KC85}) proposed a non-quantum behaviour for the CHSH Bell scenario that maximizes the violation of the CHSH Bell inequality without violating the condition of non-signalling. This behaviour, known as a pair of Popescu-Rohrlich boxes, can be expressed, using the convention introduced in Eq.~(\ref{beha}), as follows:
\begin{equation}
\label{PR}
\mathbf{p}_{\rm PR}(S_{\rm CHSH}) =
\begin{bmatrix}
\frac{1}{2} & 0 & 0 & \frac{1}{2} \\[5pt]
\frac{1}{2} & 0 & 0 & \frac{1}{2} \\[5pt]
\frac{1}{2} & 0 & 0 & \frac{1}{2} \\[5pt]
0 & \frac{1}{2} & \frac{1}{2} & 0
\end{bmatrix}.
\end{equation}
$\mathbf{p}_{\rm PR}(S_{\rm CHSH})$ violates the EP when we consider two independent copies of $\mathbf{p}_{\rm PR}$ \cite{Cabello13,FSABCLA13}. Therefore, under assumptions~\ref{Ass:BKS} and~\ref{Ass:Independence}, lemma~\ref{Lem:1} shows that a pair of statistically independent Popescu-Rohrlich boxes cannot be constructed with ideal measurements.

Moreover, a pair of statistically independent Popescu-Rohrlich boxes can be seen as a single bipartite Bell experiment. Therefore, according to lemma~\ref{Th:3}, the behaviour of this bipartite Bell experiment must satisfy the~EP. But it does not. 


\subsection{Almost quantum correlations}


In Ref.~\cite{NGHA15}, Navascu\'es {\em et al.} presented, for every Bell scenario, a set of behaviours, called almost quantum behaviours, that satisfy all principles proposed to that date but contains non-quantum behaviours. For example, the following is a non-quantum almost quantum behaviour for the CHSH Bell scenario, using the convention introduced in Eq.~(\ref{beha}):
\begin{equation}
\label{AQ}
\mathbf{p}_{\rm AQ}(S_{\rm CHSH})=
\begin{bmatrix}
\frac{2993}{5500} & \frac{8}{1375} & \frac{137}{500} & \frac{22}{125} \\[5pt]
\frac{107}{700} & \frac{139}{350} & \frac{139}{350} & \frac{37}{700} \\[5pt]
\frac{7}{11}+\frac{\sqrt{2}}{9} & \frac{2}{11}-\frac{\sqrt{2}}{9} & \frac{2}{11}-\frac{\sqrt{2}}{9} & \frac{\sqrt{2}}{9} \\[5pt]
\frac{2993}{5500} & \frac{137}{500} & \frac{8}{1375} & \frac{22}{125}
\end{bmatrix}.
\end{equation}

For a given Bell scenario $S$, the set of almost quantum correlations is a subset of $\thb(G_S)$ that satisfies the normalization and non-disturbance constraints for $S$ [constraints~(A) and (B) in Sec.~\ref{sec:problem}] \cite{NGHA15}. However, our result implies that every non-quantum behaviour in the set of almost quantum correlations for $S$ fails to satisfy constraint~(C) for $S$. For the case of $\mathbf{p}_{\rm AQ}(S_{\rm CHSH})$ in (\ref{AQ}), this failure can be checked with a computer: $\mathbf{p}_{\rm AQ}(S_{\rm CHSH})$ cannot be attained if we demand that $P(ab|xy)=|\langle \psi' | \psi \rangle|^2$, with $|\psi' \rangle = N_{ab|xy} E^{(x)}_a E^{(y)}_b |\psi \rangle$, where $N_{ab|xy}$ is a normalisation constant, and $E^{(x)}_a$ and $E^{(y)}_b$ are projection operators corresponding to, respectively, measurement $x$ with outcome $a$ and measurement $y$ with outcome $b$, and such that, for all $x$, 
$E^{(x)}_{a} E^{(x)}_{a'} = \delta_{a,a'} E^{(x)}_{a}$ and 
$\sum_{a \in A} E^{(x)}_a=\id$, for all $y$, $E^{(y)}_{b} E^{(y)}_{b'} = \delta_{b,b'} E^{(y)}_{y}$ and $\sum_{b \in B} E^{(y)}_b=\id$, and, for all $x,y$, $[E^{(x)}_a, E^{(y)}_b]=0$.


\subsection{Real versus complex quantum theory}


Theorem~\ref{Th:scenario} explains the origin of some of the most mysterious rules of quantum theory. However, it does not provide a full reconstruction of quantum theory. A natural question then is what else is needed to recover the whole formalism. 

A first specific target follows from the observation that quantum correlations for Bell and KS scenarios admit representations both in real and complex vector spaces \cite{Stueckelberg59,Stueckelberg60,MMG09,ABW13}, while the full quantum theory is formulated in {\em complex} vector spaces. Why is so?

In previous operational reconstructions of quantum theory, the complex-vector-space representation is enforced by imposing the axiom of `local tomography' (the state of a composite system is determined by the joint probabilities it assigns to measurement outcomes on the component subsystems) \cite{Hardy01,MM11,CDP11,BMU14}. Instead of this axiom, a purely physical reason enforcing the complex-vector-space representation could be the requirement that the laws of nature stay the same for all observers that are moving with respect to one another within an inertial frame. This suggests that the theory must be Lorentz-invariant at an ontological level. But this requires the theory to be free of instantaneous influences and holistic spaces inaccessible to the experimenters. In contrast to that, real-vector-space representations of quantum correlations seem to require \cite{ABW13} a holistic inaccessible ontological space. This observation suggests a way to attack the problem of why the complex-vector-space representation that is worth further examination.


\subsection{Quantum logic and the exclusivity principle}


As pointed out by Henson in Ref.~\cite{Henson15}, the EP is related to the notion of orthocoherence that appeared in quantum logic. `[A]n orthoalgebra is orthocoherent if and only if finite pairwise summable subsets are jointly summable' \cite{Wilce17}. The origin of orthocoherence can be traced back to Mackey's axiom~V in Ref.~\cite{Mackey63}. Interestingly, in the literature of quantum logic, since the end of the 1970's, orthocoherence is presented as `suspicious (\ldots) as a fundamental principle' \cite{Wilce00}. Three related reasons are offered for that:
\begin{itemize}
	\item[(a)] `There has never been given any entirely compelling reason for regarding orthocoherence as an essential feature of all reasonable physical models' \cite{Wilce17}.
	\item[(b)] `[I]t is quite easy to manufacture simple and plausible toy examples that are not orthocoherent' \cite{Wilce00} (see, e.g. \cite{Specker60,Wright78}).
	\item[(c)] `[Orthocoherence] is not stable under formation of the tensor product' \cite{Wilce00} (see \cite{RF79,FR81}).
\end{itemize}

In contrast, in this paper, we have seen that
\begin{itemize}
	\item[(a')] A compelling reason why the EP is an essential feature of any reasonable physical theory is that the EP holds for events produced by ideal measurements. Ideal measurements must exist in any theory that contains classical physics as a particular case. There are also other reasons. For example, in Ref.~\cite{BMU14}, it is shown that two postulates are sufficient to guarantee the~EP. The first postulate says that every state can be represented as a probabilistic mixture of perfectly distinguishable pure states. The second postulate is that every set of perfectly distinguishable pure states of a given dimension can be reversibly transformed to any other such set of the same dimension.
	In Ref.~\cite{Henson15}, it is shown that irreducible third-order interference (a generalization of the idea that no probabilistic interference remains unaccounted for once we have taken into account interference between pairs of slits in an $n$-slit experiment) also implies the~EP. Finally, in Ref.~\cite{CCKM19}, it is shown that every Bayesian framework must include a set of ideal experiments that must satisfy the~EP. 
	\item[(b')] As we have seen before, the toy examples in Refs.~\cite{Specker60,Wright78} (i.e. Specker's triangle and Wright's pentagon) are impossible or trivial, depending on whether or not we make the assumption that measurements are ideal. 
	\item[(c')] Observation (c) is not an obstacle but an opportunity. The fact that a behaviour $\mathbf{p}(S)$ fails to satisfy the EP applied to $n$~copies indicates that $\mathbf{p}(S)$ is forbidden under assumptions~\ref{Ass:BKS} and \ref{Ass:Independence}.
\end{itemize}


\subsection{Bohmian mechanics, many-worlds, and QBism}


It is interesting to compare the explanation of the origin of the Born rule suggested in this paper with other explanations given by some interpretations of quantum theory.

In the Bohmian interpretation \cite{Bohm52a,Bohm52b,Bohm53}, measurement outcomes are determined by a field (called the quantum potential) that permeates the universe. Each particle is at a certain position, and this position is governed by the quantum potential. However, the quantum potential is hidden to the experimenters, so experimenters are restricted to calculating the probability density to observe that the particle is at some position. The reason why outcomes happen with relative frequencies given by the Born rule is the initial state of the quantum potential and the initial positions of all particles \cite{Bohm53}.

In the Everettian interpretation of quantum theory \cite{Everett57}, a measurement is a unitary transformation of a universal wave function that gives rise to a multiverse. Everettians and proponents of the decoherent (or consistent) histories interpretation of quantum theory \cite{GH90}
have made several attempts to justify the Born rule from simpler assumptions (e.g. \cite{Finkelstein65,Hartle68,FGG89,Deutsch99,Wallace03,Wallace07,Zurek03,Zurek05,Vaidman11}). However, each of these attempts in turn has attracted critical attention (e.g. \cite{Squires90,CS96,BCFFS00,CS05,SF05}).

The problem is that any derivation of the Born rule that {\em assumes} that measurements can be represented by self-adjoint operators in a vector space ${\mathcal V}$, outcomes correspond to their eigenvalues, and mutually exclusive results correspond to orthonormal basis is almost circular. Because then, for every ${\mathcal V}$ of dimension $d$ greater than two, Gleason's theorem \cite{Gleason57} shows that the only possible states are vectors in ${\mathcal V}$ and the only possible choice to assign a probability $P(v_i)$ to every vector $v_i$ in ${\mathcal V}$ for any given vector $v$ in ${\mathcal V}$ such that, for every orthonormal basis $\{v_i\}_{i=1}^d$ of ${\mathcal V}$, the sum of the probabilities satisfies $\sum_{i=1}^d P(v_i)=1$ and $P(v_i)$ only depends on $v_i$ (and not on the orthonormal basis of ${\mathcal V}$ considered) is $P(v_i) = | \langle v_i | v \rangle |^2$. 

QBism \cite{FS13} is a different type of interpretation. It does not assume any particular picture of the universe. For QBism, quantum theory is a {\em personal} tool for each agent. According to QBism, the Born rule is not a law of nature in the usual sense, but `an empirical {\em addition} to the laws of Bayesian probability' \cite{FS13} that a wise agent should follow in addition to the Bayesian coherence conditions. However, QBism does not clarify where does this empirical addition come from. In this respect, Theorem~\ref{Th:scenario} could be an important contribution to QBism as it shows that the set of quantum behaviours does not require any empirical addition to the laws of Bayesian probability but follows from two natural assumptions that fit perfectly within a Bayesian framework.


\subsection{Emergence of classicality}


Existing explanations of the emergence of classicality from quantum theory such as decoherence \cite{Zurek91} and the restriction to coarse-grained measurements \cite{KB07} assume the formalism of quantum theory. An interesting problem is understanding the emergence of classicality within the framework proposed in this paper. An interesting observation in this sense is that, for any self-complementary graph of exclusivity $G$, the~EP applied to a single copy selects a set, ${\cal E}^{1}(G)$ (see Sec.~\ref{sec:proof}) whose elements do not violate the~EP with other elements if and only if these other elements belong to $\abl[{\cal E}^{1}(G)]$, which is exactly the set of classical assignments of probability for~$G$ \cite{CSW14}. This points out that the set of classical assignments of probability for $G$ is particularly robust under the EP and suggests that the problem of the emergence of classicality may be benefit from the approach to quantum theory proposed in this paper.\vskip6pt

\enlargethispage{20pt}


\competing{The author declares that he has no competing interests.}


\funding{This work was supported by the Foundational Questions Institute (FQXi) Large Grant `The Observer Observed: A Bayesian Route to the Reconstruction of Quantum Theory' (FQXi-RFP-1608). The author is also supported by the MINECO-MCINN project `Quantum Tools for Information, Computation and Research' (FIS2017-89609-P) with FEDER funds and the Knut and Alice Wallenberg Foundation project `Photonic Quantum Information'. This study has been partially financed by the Conserjer\'{\i}a de Conocimiento, Investigaci\'on y Universidad, Junta de Andaluc\'{\i}a and European Regional Development Fund (ERDF), Ref.~SOMM17/6105/UGR.}


\ack{The author thanks B.~Amaral, M.~Ara\'ujo, C.~Budroni, E.~Cavalcanti, G.~Chiribella, P.~Grangier, M.~J.~W.~Hall, G.~Jaegger, M.~Kleinmann, S.~L\'opez-Rosa, A.~J.~L\'opez-Tarrida, M.~P.~M\"uller, A.~Ribeiro~de~Carvalho, M.~Terra~Cunha, H.~M.~Wiseman, Z.-P.~Xu and W.~H.~Zurek for discussions and comments.}




\begin{thebibliography}{9}
	

\bibitem{Wheeler88}
Wheeler JA. 1988
In 
\href{https://doi.org/10.1007/978-94-009-3061-2_7}{{\em Probability in the sciences} (ed. E Agazzi), p.~103. Dordrecht, Holland: Springer.}

\bibitem{Wheeler83a}
Wheeler JA. 1983
`On recognizing ``law without law,''' Oersted medal response at the joint APS-AAPT Meeting, New York, 25 January 1983.
\href{https://doi.org/10.1119/1.13224}{{\em Am. J. Phys.} \textbf{51}, 398--404.}

\bibitem{MTZ09}
Misner CW, Thorne KS, Zurek WH. 2009
John Wheeler, relativity, and quantum information.
\href{https://doi.org/10.1063/1.3120895}{{\em Phys. Today} \textbf{62}, 40--46.}


\bibitem{Gell-Mann56}
Gell-Mann M. 1956
The interpretation of the new particles as displaced charge multiplets.
\href{https://doi.org/10.1007/BF02748000}{{\em Il Nuovo Cimento} \textbf{4}, 848--866.}


\bibitem{Bell64}
Bell JS. 1964
On the Einstein Podolsky Rosen paradox.
\href{https://doi.org/10.1103/PhysicsPhysiqueFizika.1.195}{{\em Physics} \textbf{1}, 195--200.}

\bibitem{CHSH69}
Clauser JF, Horne MA, Shimony A, Holt RA. 1969
Proposed experiment to test local hidden-variable theories.
\href{https://doi.org/10.1103/PhysRevLett.23.880}{{\em Phys. Rev. Lett.} \textbf{23}, 880--884.}

\bibitem{Specker60}
Specker EP. 1960
Die Logik nicht gleichzeitig entscheidbarer Aussagen.
\href{https://doi.org/10.1111/j.1746-8361.1960.tb00422.x}{{\em Dialectica} \textbf{14}, 239--246.} 

\bibitem{Bell66}
Bell JS. 1966
On the problem of hidden variables in quantum mechanics.
\href{https://doi.org/10.1103/RevModPhys.38.447}{{\em Rev. Mod. Phys.} \textbf{38}, 447--452.}

\bibitem{KS67}
Kochen S, Specker EP. 1967
The problem of hidden variables in quantum mechanics.
\href{https://doi.org/10.1512/iumj.1968.17.17004}{{\em J. Math. Mech.} \textbf{17}, 59--87.}


\bibitem{Cabello18}
Cabello A. 2019
Quantum correlations from simple assumptions.
\href{https://doi.org/10.1103/PhysRevA.100.032120}{{\em Phys. Rev. A} \textbf{100}, 032120.}


\bibitem{CY14}
Chiribella G, Yuan X.
Measurement sharpness cuts nonlocality and contextuality in every physical theory.
\href{http://arxiv.org/abs/1404.3348}{arXiv:1404.3348.}


\bibitem{Gisin91}
Gisin N. 1991
Bell's inequality holds for all non-product states.
\href{https://doi.org/10.1016/0375-9601(91)90805-I}{{\em Phys. Lett. A} \textbf{154}, 201--202.}

\bibitem{KCBS08}
Klyachko AA, Can MA, Binicio\u{g}lu S, Shumovsky AS. 2008 
Simple test for hidden variables in spin-1 systems.
\href{https://doi.org/10.1103/PhysRevLett.101.020403}{{\em Phys. Rev. Lett.} \textbf{101}, 020403.}

\bibitem{Cabello08}
Cabello A. 2008
Experimentally testable state-independent quantum contextuality.
\href{https://doi.org/10.1103/PhysRevLett.101.210401}{{\em Phys. Rev. Lett.} \textbf{101}, 210401.}

\bibitem{BBCP09}
Badzi\c{a}g P, Bengtsson I, Cabello A, Pitowsky I. 2009
Universality of state-independent violation of correlation inequalities for noncontextual theories.
\href{https://doi.org/10.1103/PhysRevLett.103.050401}{{\em Phys. Rev. Lett.} \textbf{103}, 050401.}

\bibitem{YO12}
Yu S, Oh CH. 2012
State-independent proof of Kochen-Specker theorem with 13 rays.
\href{https://doi.org/10.1103/PhysRevLett.108.030402}{{\em Phys. Rev. Lett.} \textbf{108}, 030402.}

\bibitem{KBLGC12}
Kleinmann M, Budroni C, Larsson J-\AA, G{\"u}hne O, Cabello A. 2012
Optimal inequalities for state-independent contextuality.
\href{https://doi.org/10.1103/PhysRevLett.109.250402}{{\em Phys. Rev. Lett.} \textbf{109}, 250402.}


\bibitem{Spekkens14}
Spekkens RW. 2014
The status of determinism in proofs of the impossibility of a noncontextual model of quantum theory.
\href{https://doi.org/10.1007/s10701-014-9833-x}{{\em Found. Phys.} \textbf{44}, 1125--1155.}


\bibitem{Peres91}
Peres A. 1991
Two simple proofs of the Kochen-Specker theorem.
\href{https://doi.org/10.1088/0305-4470/24/4/003}{{\em J. Phys. A: Math. Gen.} \textbf{24}, L175--L178.}

\bibitem{CEG96}
Cabello A, Estebaranz JM, Garc\'{\i}a-Alcaine G. 1996
Bell-Kochen-Specker theorem: a proof with 18 vectors.
\href{https://doi.org/10.1016/0375-9601(96)00134-X}{{\em Phys. Lett. A} \textbf{212}, 183--187.}

\bibitem{LBPC14}
Lison\v{e}k P, Badzi\c{a}g P, Portillo JR, Cabello A. 2014
Kochen-Specker set with seven contexts.
\href{https://doi.org/10.1103/PhysRevA.79.012102}{{\em Phys. Rev. A} \textbf{89}, 042101.}


\bibitem{BBC12}
Bengtsson I, Blanchfield K, Cabello A. 2012
A Kochen-Specker inequality from a SIC.
\href{http://doi.org/10.1016/j.physleta.2011.12.011}{{\em Phys. Lett. A} \textbf{376}, 374--376.}


\bibitem{CSW10}
Cabello A, Severini S, Winter A.
(Non-)Contextuality of physical theories as an axiom.
\href{https://arxiv.org/abs/1010.2163}{arXiv:1010.2163.}

\bibitem{CSW14}
Cabello A, Severini S, Winter A. 2014
Graph-theoretic approach to quantum correlations.
\href{https://doi.org/10.1103/PhysRevLett.112.040401}{{\em Phys. Rev. Lett.} \textbf{112}, 040401.}


\bibitem{Tsirelson93}
Tsirelson BS. 1993
Some results and problems on quantum Bell-type inequalities.
{\em Hadronic J. Suppl.} \textbf{8}, 329--345.

\bibitem{BCPSW14}
Brunner N, Cavalcanti D, Pironio S, Scarani V, Wehner S. 2014
Bell nonlocality.
\href{https://doi.org/10.1103/RevModPhys.86.419}{{\em Rev. Mod. Phys.} \textbf{86}, 419--478.}

\bibitem{AB11}	
Abramsky S, Brandenburger A. 2011	
The sheaf-theoretic structure of non-locality and contextuality.	
\href{https://doi.org/10.1088/1367-2630/13/11/113036}{{\em New J. Phys.} \textbf{13}, 113036.}

\bibitem{AH12}
Abramsky S, Hardy L. 2012
Logical Bell inequalities.
\href{https://doi.org/10.1103/PhysRevA.85.062114}{{\em Phys. Rev. A} \textbf{85}, 062114.}

\bibitem{AFLS15}
Ac\'{\i}n A, Fritz T, Leverrier A, Sainz AB. 2015
A combinatorial approach to nonlocality and contextuality.
\href{https://doi.org/10.1007/s00220-014-2260-1}{{\em Commun. Math. Phys.} \textbf{334}, 533--628.}


\bibitem{Neumark40a}
Neumark MA. 1940
On the self-adjoint extensions of the second kind of a symmetric operator.
{\em Izv. Akad. Nauk S.S.S.R. [Bull. Acad. Sci. U.S.S.R.] S\'er. Mat.} \textbf{4}, 53--104 (Russian with English summary).

\bibitem{Neumark40b}
Neumark MA. 1940
Spectral functions of a symmetric operator.
{\em Izv. Akad. Nauk S.S.S.R. [Bull. Acad. Sci. U.S.S.R.] S\'er. Mat.} \textbf{4}, 277--318 (Russian with English summary).

\bibitem{Neumark43}
Neumark MA. 1943
On a representation of additive operator set functions.
{\em C.R. (Dokl.) Acad. Sci. U.R.S.S. (N.S.)} \textbf{41}, 359--361.

\bibitem{Holevo80}
Holevo AS. 2011
{\em Probabilistic and statistical aspects of quantum theory}, p.~55. 
Pisa, Italy: Scuola Normale Superiore Pisa.
First published in Russian in 1980.

\bibitem{Peres95}
Peres A. 1995
{\em Quantum theory: concepts and methods}, p.~285.
New York, NY: Kluwer.


\bibitem{Wigner60}
Wigner EP. 1960 
The unreasonable effectiveness of mathematics in the natural sciences. Richard Courant lecture in mathematical sciences delivered at New York University, May 11, 1959.
\href{https://doi.org/10.1002/cpa.3160130102}{{\em Commun. Pure Appl. Math.} \textbf{13}, 1--14.} 


\bibitem{PR94}
Popescu S, Rohrlich D. 1994
Quantum nonlocality as an axiom.
\href{https://doi.org/10.1007/BF02058098}{{\em Found. Phys.} \textbf{24}, 379--385.}

\bibitem{vanDam99}
van Dam W. 1999
Nonlocality \& communication complexity.
PhD thesis, University of Oxford, Oxford, UK.

\bibitem{PPKSWZ09}
Paw{\l}owski M, Paterek T, Kaszlikowski D, Scarani V, Winter A, \.{Z}ukowski M. 2009
Information causality as a physical principle.
\href{https://doi.org/10.1038/nature08400}{{\em Nature} \textbf{461}, 1101--1104.}

\bibitem{NW09}
Navascu\'es M, Wunderlich H. 2009
A glance beyond the quantum model.
\href{https://doi.org/10.1098/rspa.2009.0453}{{\em Proc. R. Soc. A} \textbf{466}, 881--890.}

\bibitem{Cabello13}
Cabello A. 2013
Simple explanation of the quantum violation of a fundamental inequality.
\href{https://doi.org/10.1103/PhysRevLett.110.060402}{{\em Phys. Rev. Lett.} \textbf{110}, 060402.}

\bibitem{FSABCLA13}
Fritz T, Sainz AB, Augusiak R, Bohr Brask J, Chaves R, Leverrier A, Ac\'{\i}n A. 2013
Local orthogonality as a multipartite principle for quantum correlations.
\href{https://doi.org/10.1038/ncomms3263}{{\em Nat. Commun.} \textbf{4}, 2263.}


\bibitem{NGHA15}
Navascu\'es M, Guryanova Y, Hoban MJ, Ac\'{\i}n A. 2015
Almost quantum correlations.
\href{https://doi.org/10.1038/ncomms7288}{{\em Nat. Commun.} \textbf{6}, 6288.}


\bibitem{BS69}
Barnard ACL, Sallin EA. 1969
Tachyons and tardyons.
\href{https://doi.org/10.1063/1.3035237}{{\em Phys. Today} \textbf{22}, 9.}

\bibitem{Zweig10}
Zweig G. 2010
Memories of Murray and the quark model.
\href{https://doi.org/10.1142/S0217751X10050494}{{\em Int. J. Mod. Phys. A} \textbf{25}, 3863--3877.}

\bibitem{Trigg70}
Trigg GL. 1970
Tachyons revisited.
\href{https://doi.org/10.1063/1.3021812}{{\em Phys. Today} \textbf{23}, 79.}

\bibitem{Weinberg05}
Weinberg S. 2005
Einstein's mistakes.
\href{https://doi.org/10.1063/1.2155755}{{\em Phys. Today} \textbf{58}, 31--35.}

\bibitem{White38}
White TH. 1938
{\em The sword in the stone}.
London, UK: Collins.


\bibitem{Yan13}
Yan B. 2013
Quantum correlations are tightly bound by the exclusivity principle.
\href{https://doi.org/10.1103/PhysRevLett.110.260406}{{\em Phys. Rev. Lett.} \textbf{110}, 260406.}

\bibitem{ATC14}
Amaral B, Terra Cunha M, Cabello A. 2014
Exclusivity principle forbids sets of correlations larger than the quantum set.
\href{https://doi.org/10.1103/PhysRevA.89.030101}{{\em Phys. Rev. A} \textbf{89}, 030101(R).}

\bibitem{Cabello15}
Cabello A. 2015
Simple explanation of the quantum limits of genuine $n$-body nonlocality.
\href{https://doi.org/10.1103/PhysRevLett.114.220402}{{\em Phys. Rev. Lett.} \textbf{114}, 220402.}

\bibitem{Henson15}
Henson J. 2015
Bounding quantum contextuality with lack of third-order interference.
\href{https://doi.org/10.1103/PhysRevLett.114.220403}{{\em Phys. Rev. Lett.} \textbf{114}, 220403.}


\bibitem{GLS88}
Gr\"otschel M, Lov\'asz L, Schrijver A. 1988
{\em Geometric algorithms and combinatorial optimization}.
Berlin, Germany: Springer.

\bibitem{Knuth94}
Knuth DE. 1994
The sandwich theorem.
\href{http://www.combinatorics.org/ojs/index.php/eljc/article/view/v1i1a1}{{\em Electron. J. Combin.} \textbf{1}, A1.}

\bibitem{GLS86}
Gr\"otschel M, Lov\'asz L, Schrijver A. 1986
Relaxations of vertex packing.
\href{https://doi.org/10.1016/0095-8956(86)90087-0}{{\em J. Combin. Theory B} \textbf{40}, 330--343.}


\bibitem{Born26}
Born M. 1926
Zur Quantenmechanik der Sto{\ss}vorg\"ange.
\href{https://doi.org/10.1007/BF01397477}{{\em Z. Physik} \textbf{37}, 863--867.} 
 

\bibitem{LSW11}
Liang Y-C, Spekkens RW, Wiseman HM. 2011
Specker's parable of the overprotective seer: a road to contextuality, nonlocality and complementarity.
\href{https://doi.org/10.1016/j.physrep.2011.05.001}{{\em Phys. Rep.} \textbf{506}, 1--39.}

\bibitem{LSW17}
Liang Y-C, Spekkens RW, Wiseman HM. 2017
Erratum to `Specker's parable of the overprotective seer: a road to contextuality, nonlocality and complementarity' [Phys. Rep. 506 (2011) 1--39].
\href{https://doi.org/10.1016/j.physrep.2016.12.001}{{\em Phys. Rep.} \textbf{666}, 110--111.}

\bibitem{Wright78}
Wright R. 1978
In \emph{Mathematical foundations of quantum mechanics} (ed. AR Marlow), p.~255.
San Diego, CA: Academic Press.


\bibitem{Rastall85}
Rastall P. 1985
Locality, Bell's theorem, and quantum mechanics.
\href{https://doi.org/10.1007/BF00739036}{{\em Found. Phys.} \textbf{15}, 963--972.}

\bibitem{KC85}
Khalfin LA, Tsirelson BS. 1985
In
{\em Symposium on the foundations of modern physics: 50 years of the Einstein-Podolsky-Rosen experiment}
(eds PJ Lahti, P Mittelstaedt), p.~441.
Singapore: World Scientific. 


\bibitem{Stueckelberg59}
Stueckelberg ECG. 1959 
Field quantisation and time reversal in real Hilbert space.
{\em Helv. Phys. Acta} \textbf{32}, 254--256.

\bibitem{Stueckelberg60}
Stueckelberg ECG. 1960
Quantum theory in real Hilbert space.
{\em Helv. Phys. Acta} \textbf{33}, 727--752.

\bibitem{MMG09}
McKague M, Mosca M, Gisin N. 2009
Simulating quantum systems using real Hilbert spaces.
\href{https://doi.org/10.1103/PhysRevLett.102.020505}{{\em Phys. Rev. Lett.} \textbf{102}, 020505.}

\bibitem{ABW13}
Aleksandrova A, Borish V, Wootters WK. 2013
Real-vector-space quantum theory with a universal quantum bit.
\href{https://doi.org/10.1103/PhysRevA.87.052106}{{\em Phys. Rev. A} \textbf{87}, 052106.}


\bibitem{Hardy01}
Hardy L.
Quantum theory from five reasonable axioms.
\href{http://arxiv.org/abs/quant-ph/0101012}{quant-ph/0101012.}

\bibitem{MM11} 
Masanes L, M\"uller MP. 2011 
A derivation of quantum theory from physical requirements. 
\href{https://doi.org/10.1088/1367-2630/13/6/063001}{{\em New J. Phys.} \textbf{13}, 063001.}

\bibitem{CDP11}
Chiribella G, D'Ariano GM, Perinotti P. 2011
Informational derivation of quantum theory.
\href{https://doi.org/10.1103/PhysRevA.84.012311}{{\em Phys. Rev. A} \textbf{84}, 012311.}

\bibitem{BMU14}
Barnum H, M\"uller MP, Ududec C. 2014
Higher-order interference and single-system postulates characterizing quantum theory.
\href{https://doi.org/10.1088/1367-2630/16/12/123029}{{\em New J. Phys.} \textbf{16}, 123029.}


\bibitem{Wilce17}
Wilce A. 2017
Quantum logic and probability theory.
In
\href{https://plato.stanford.edu/archives/spr2017/entries/qt-quantlog/}{{\em The Stanford encyclopedia of philosophy} (ed. EN Zalta).}
Stanford, CA: Metaphysics Research Lab.

\bibitem{Mackey63}
Mackey GW. 1963
{\em The mathematical foundations of quantum theory}.
New York, NY: W. A. Benjamin.

\bibitem{Wilce00}
Wilce A. 2000 
In
\href{https://doi.org/10.1007/978-94-017-1201-9_4}{{\em Current research in operational quantum logic}
	(eds B Coecke, D Moore, A Wilce), p.~81. Amsterdam, Netherlands: Springer.}

\bibitem{RF79}
Randall CH, Foulis DJ. 1979
Tensor products of quantum logics do not exist. 
{\em Notices Am. Math. Soc.} \textbf{26}, A-557.

\bibitem{FR81}
Foulis DJ, Randall, CH. 1981
In {\em Interpretations and foundations of quantum mechanics}
(ed. H Neumann), p.~9. Mannheim, Germany: Wissenchaftsverlag. 


\bibitem{CCKM19} 
Chiribella G, Cabello A, Kleinmann M, M\"uller, MP.
General Bayesian theories and the emergence of the exclusivity principle.
\href{https://arxiv.org/abs/1901.11412}{arXiv:1901.11412.}


\bibitem{Bohm52a}
Bohm D. 1952
A suggested interpretation of the quantum theory in terms of `hidden' variables. I.
\href{https://doi.org/10.1103/PhysRev.85.166}{{\em Phys. Rev.} \textbf{85}, 166--179.}

\bibitem{Bohm52b}
Bohm D. 1952
A suggested interpretation of the quantum theory in terms of `hidden' variables. II.
\href{https://doi.org/10.1103/PhysRev.85.180}{{\em Phys. Rev.} \textbf{85}, 180--193.}

\bibitem{Bohm53}
Bohm D. 1953
Proof that probability density approaches $|\psi|^2$ in causal interpretation of quantum theory.
\href{https://doi.org/10.1103/PhysRev.89.458}{{\em Phys. Rev.} \textbf{89}, 458--466.}


\bibitem{Everett57}
Everett III H. 1957
`Relative state' formulation of quantum mechanics.
\href{https://doi.org/10.1103/RevModPhys.29.454}{{\em Rev. Mod. Phys.} \textbf{29}, 454--462.}

\bibitem{GH90}
Gell-Mann M, Hartle JB. 1990
In {\em Entropy, and the physics of information} (ed. WH Zurek). SFI Studies in the Sciences of Complexity, vol.~VIII p.~425. Reading, MA: Addison-Wesley. 


\bibitem{Finkelstein65}
Finkelstein D. 1965 
Section of physical sciences: the logic of quantum physics.
\href{https://doi.org/10.1111/j.2164-0947.1963.tb01483.x}{{\em Trans. NY Acad. Sci.} \textbf{25}, 621--637.}

\bibitem{Hartle68}
Hartle JB. 1968
Quantum mechanics of individual systems.
\href{https://doi.org/10.1119/1.1975096}{{\em Am. J. Phys.} \textbf{36}, 704--712.}

\bibitem{FGG89}
Farhi E, Goldstone J, Gutmann S. 1989 
How probability arises in quantum mechanics.
\href{https://doi.org/10.1016/0003-4916(89)90141-3}{{\em Ann. Phys.} \textbf{192}, 368--382.}

\bibitem{Deutsch99}
Deutsch D. 1999
Quantum theory of probability and decisions.
\href{https://doi.org/10.1098/rspa.1999.0443}{{\em Proc. R. Soc. Lond. A} \textbf{455}, 3129--3137.}

\bibitem{Wallace03}
Wallace D. 2003 
Everettian rationality: defending Deutsch's approach to probability in the Everett interpretation.
\href{https://doi.org/10.1016/S1355-2198(03)00036-4}{{\em Stud. Hist. Phil. Sci. B} \textbf{34}, 415--439.}

\bibitem{Wallace07}
Wallace D. 2007 
Quantum probability from subjective likelihood: improving on Deutsch's proof of the probability rule.
\href{https://doi.org/10.1016/j.shpsb.2006.04.008}{{\em Stud. Hist. Phil. Sci. B} \textbf{38}, 311--332.}

\bibitem{Zurek03}
Zurek WH. 2003 
Environment-assisted invariance, entanglement, and probabilities in quantum physics.
\href{https://doi.org/10.1103/PhysRevLett.90.120404}{{\em Phys. Rev. Lett.} \textbf{90}, 120404.}

\bibitem{Zurek05}
Zurek WH. 2005 
Probabilities from entanglement, Born's rule $p_k=|\psi_k|^2$ from envariance.
\href{https://doi.org/10.1103/PhysRevA.71.052105}{{\em Phys. Rev. A} \textbf{71}, 052105.}

\bibitem{Vaidman11}
Vaidman L. 2012 
In 
\href{https://doi.org/10.1007/978-3-642-21329-8_18}{{\em Probability in physics} (eds Y Ben-Menahem, M Hemmo ), p.~299. Berlin, Germany: Springer.}


\bibitem{Squires90}
Squires EJ. 1990
On an alleged `proof' of the quantum probability law.
\href{https://doi.org/10.1016/0375-9601(90)90192-Q}{{\em Phys. Lett. A} \textbf{145}, 67--68.}

\bibitem{CS96}
Cassinello A, S\'anchez-G\'omez JL. 1996
On the probabilistic postulate of quantum mechanics.
\href{https://doi.org/10.1007/BF02058273}{{\em Found. Phys.} \textbf{26}, 1357--1374.}

\bibitem{BCFFS00}
Barnum H, Caves CM, Finkelstein J, Fuchs CA, Schack R. 2000
Quantum probability from decision theory?
\href{https://doi.org/10.1098/rspa.2000.0557}{{\em Proc. R. Soc. Lond. A} \textbf{456}, 1175--1182.}

\bibitem{CS05}
Caves CM, Schack R. 2005
Properties of the frequency operator do not imply the quantum probability postulate.
\href{https://doi.org/10.1016/j.aop.2004.09.009}{{\em Ann. Phys.} \textbf{315}, 123--146.}

\bibitem{SF05}
Schlosshauer M, Fine, A. 2005
On Zurek's derivation of the Born rule.
\href{https://doi.org/10.1007/s10701-004-1941-6}{{\em Found. Phys.} \textbf{35}, 197--213.}


\bibitem{Gleason57}
Gleason, AM. 1957
Measures on the closed subspaces of a Hilbert space.
\href{https://doi.org/10.1512/iumj.1957.6.56050}{{\em J. Math. Mech.} \textbf{6}, 885--893.}


\bibitem{FS13}
Fuchs CA, Schack R. 2013
Quantum-Bayesian coherence.
\href{https://doi.org/10.1103/RevModPhys.85.1693}{{\em Rev. Mod. Phys.} \textbf{85}, 1693--1715.}


\bibitem{Zurek91}
Zurek WH. 1991
Decoherence and the transition from quantum to classical.
\href{https://doi.org/10.1063/1.881293}{{\em Phys. Today} \textbf{44}, 36--44.} 

\bibitem{KB07}
Kofler J, Brukner \v{C}. 2007
Classical world arising out of quantum physics under the restriction of coarse-grained measurements.
\href{https://doi.org/10.1103/PhysRevLett.99.180403}{{\em Phys. Rev. Lett.} \textbf{99}, 180403.}

\end{thebibliography}
\end{document}